\documentclass[journal]{IEEEtran}
\usepackage{amsmath,amssymb,amsfonts,amsthm,amscd}
\usepackage{graphicx}
\usepackage{booktabs}
\usepackage{xspace}
\usepackage{tikz}
\usetikzlibrary{automata,positioning,arrows.meta}
\usepackage{url}
\usepackage{microtype}

\theoremstyle{plain}
\newtheorem{theorem}{Theorem}
\newtheorem{proposition}[theorem]{Proposition}
\newtheorem{lemma}[theorem]{Lemma}
\newtheorem{corollary}[theorem]{Corollary}
\theoremstyle{definition}
\newtheorem{definition}[theorem]{Definition}
\newtheorem{problem}[theorem]{Problem}
\theoremstyle{remark}
\newtheorem{example}[theorem]{Example}
\newtheorem{remark}[theorem]{Remark}

\newcommand{\So}{\Sigma_{o}}
\newcommand{\Suo}{\Sigma_{uo}}
\newcommand{\To}{T_{o}}
\newcommand{\Tuo}{T_{uo}}
\newcommand{\PT}{P_{T}}
\newcommand{\Ro}{R_{o}}
\DeclareMathOperator{\pair}{pr}
\newcommand{\Jm}{\mathcal{J}}
\newcommand{\Fm}{\mathcal{F}}
\newcommand{\PTIME}{\textsc{P}\xspace}
\newcommand{\PSPACE}{\textsc{PSpace}\xspace}
\newcommand{\BPP}{\textsc{BPP}\xspace}
\newcommand{\QEDbox}{\hfill$\square$}

\begin{document}

\title{Deciding Relabeling Observation Consistency\\ in Multi-Agent Discrete-Event Systems}

\author{Tom\'{a}\v{s}~Masopust%
\thanks{T. Masopust is with the Faculty of Science, Palack\'y University Olomouc, Czechia (e-mail: tomas.masopust@upol.cz).}}

\maketitle

\begin{abstract}
Scalable supervisors for multi-agent discrete-event systems control groups of isomorphic agents through a common template. Under partial observation, such a supervisor is maximally permissive if the relabeling that maps the agents onto their template is \emph{relabeling observation consistent} (ROC) and a local companion condition holds. Whether ROC is decidable was open. We show that it is \textsc{PSpace}-complete: for nondeterministic plants with two observable and one unobservable event, and for deterministic plants with two observable and two unobservable events. On the tractable side, we characterize the relabelings that guarantee ROC for every plant, we give a polynomial-time algorithm for deterministic plants whose relabeling is injective on unobservable events, and we give sufficient conditions based on saturation and on simulation, with polynomial-time tests for both. We further show that ROC is compositional for components with pairwise disjoint templates, with equivalence if the alphabets are moreover pairwise disjoint, and that neither hypothesis can be dropped. Finally, we show that the structural condition proposed in the literature to guarantee ROC is incorrect, and repair the companion inclusion that the same proposition asserts.
\end{abstract}

\begin{IEEEkeywords}
Discrete-event systems, Multi-agent systems, Supervisory control, Partial observation, Maximal permissiveness, Computational complexity
\end{IEEEkeywords}

\section{Introduction}\label{sec:intro}
	In multi-agent discrete-event systems (DES)~\cite{LCL19,LKL22}, several groups of isomorphic agents---machines, robots, or vehicles instantiated from a common template---run in parallel. Such systems arise across manufacturing and logistics, for instance as machines grouped by the type of workpiece they process, or as automated guided vehicles transporting items of distinct kinds~\cite{LCL19}. Their defining feature is that the number of agents is not fixed a~priori and changes over time, as units are added to raise throughput or removed after a fault; a monolithic supervisor, whose state space grows with the number of agents, has to be recomputed after every such change.

	Template- and symmetry-based approaches exploit the structural similarity of agents and include the verification and control of interacting agents~\cite{RL06,LM26}, broadcasting-based composition~\cite{Su13}, control-protocol synthesis from agent and requirement templates~\cite{SL17}, symmetry reduction by state-tree structures and by relabeling and reconfiguration~\cite{JGXW17,JGXW20}, and modular synthesis by similarity~\cite{LKML22}; see~\cite{FG24} for a recent survey.

	Among these, the framework of Liu \emph{et al.}~\cite{LKL22,LCL19} yields a scalable supervisor under partial observation whose state space and computational cost are independent of the number of agents. A \emph{relabeling} $R$ maps the events of isomorphic agents onto common \emph{template events}, collapsing each group onto one generator; for instance, two events $a_1,a_2$ of two isomorphic machines are instantiations of a single template event $\hat a=R(a_1)=R(a_2)$. The supervisor is synthesized for a \emph{template plant}, assembled from a fixed number of agents of each group, and is inverse-relabeled back to the plant level. Distinct groups have disjoint alphabets and disjoint template alphabets, and therefore the collapse happens inside the groups; it is what makes $R$ noninjective, and it is the source of the difficulties studied here.

	Safety of the scalable supervisor is guaranteed by relative observability~\cite{CZW15,ACB17} on the template, but its permissiveness is not automatic: the supervisor is synthesized against the \emph{template} rather than against the relabeled plant. Maximal permissiveness---the guarantee that the supervisor is as permissive as the monolithic one---requires, in addition to a local companion condition, that the relabeling be \emph{relabeling observation consistent} (ROC) with respect to the plant and the natural projection~\cite{LKL22}. Informally, ROC requires that whenever a template string looks, on the observable level, like the relabeling of a plant string, the plant contains a string with that relabeling and the same observation; the supervisor then loses nothing by reasoning on the template.

	The question of whether ROC is decidable was open. On the practical side, the only structural condition proposed so far asserts that ROC holds whenever the agents within each group have pairwise disjoint alphabets~\cite{LKL22,LKL21}; this assertion fails, as we show in Example~\ref{ex:sef}.

\subsubsection*{Contributions}
	We settle the complexity of ROC verification: it is \PSPACE-complete (Theorem~\ref{thm:complete}), already for automata with two observable events and one unobservable event (Lemma~\ref{lem:nfa-hard}), and for deterministic automata with two observable and two unobservable events (Lemma~\ref{lem:dfa-hard}). Decidability follows from an upper bound (Theorem~\ref{thm:upper}) obtained by encoding the pairs that ROC quantifies over as words over a \emph{pairing alphabet}: the encoding is a bijection (Lemma~\ref{lem:pairs}), and ROC becomes an inclusion between two nondeterministic automata with $O(n^2)$ states.

	On the tractable side, injectivity of the relabeling on \emph{observable} events, together with the cases $\So=\emptyset$ and $\Suo=\emptyset$, characterizes the situations in which ROC holds for every language (Proposition~\ref{prop:trivial-roc}); a strictly weaker sufficient condition is \emph{saturation}, the closedness of the plant under exchanging observable events with a common template image (Proposition~\ref{prop:sat}); and for deterministic automata with a relabeling injective on \emph{unobservable} events, ROC is decidable in polynomial time (Proposition~\ref{prop:ptime}). We further show that ROC is compositional: for components with pairwise disjoint templates, ROC of all components implies ROC of the composition, and if the alphabets are moreover pairwise disjoint, the two are equivalent (Proposition~\ref{prop:comp}); neither hypothesis can be dropped.

	Finally, we return to the structural condition of~\cite[Prop.~2]{LKL22}. Its first assertion is incorrect, and we refute it by a two-machine counterexample. Its second assertion, the inclusion $L(\mathbf{M})\subseteq R(L(\mathbf{G}))$, survives, but its published proof lets a noninjective morphism commute with intersection, which holds only as an inclusion; we prove the assertion and make the hypotheses it needs explicit (Corollary~\ref{cor:LMinRLG}).

\subsubsection*{Related conditions}
	ROC has a hierarchical counterpart. In hierarchical supervisory control, the abstraction is a projection onto a high-level alphabet rather than a relabeling, and the corresponding conditions are observation consistency and modified observation consistency~\cite{BKMSS11,KM20,KM23}. Their verification behaves entirely differently: it is undecidable~\cite{MOC}. The local companion conditions of both frameworks are \PSPACE-complete for nondeterministic plants and decidable in polynomial time for deterministic ones~\cite{LOCPAPER}. 

\section{Preliminaries}\label{sec:prelim}
	We assume familiarity with the basics of formal languages, supervisory control~\cite{CassandrasLafortune08,RamadgeWonham87}, and complexity theory~\cite{AB09}; \PTIME{} and \PSPACE{} denote the classes of problems decidable in polynomial time and in polynomial space.

	An \emph{alphabet} is a finite nonempty set of events. For an alphabet $\Sigma$, the set of all finite strings over $\Sigma$ is denoted by $\Sigma^*$, and the empty string by $\varepsilon$; for $w\in\Sigma^*$, $|w|$ denotes its length. We write $A\mathbin{\dot\cup}B$ for the union of disjoint sets $A$ and $B$. The \emph{prefix closure} of $L$ is $\overline{L}=\{u : uv\in L \text{ for some } v\}$, and $L$ is \emph{prefix-closed} if $L=\overline L$.

	A \emph{nondeterministic finite automaton} (NFA) is a structure $G=(Q,\Sigma,\delta,I,F)$, where $Q$ is a finite set of states, $I\subseteq Q$ is a set of initial states, $F\subseteq Q$ is a set of \emph{final} states, and $\delta\colon Q\times\Sigma\to2^Q$ is a transition function, extended to $2^Q\times\Sigma^*$ as usual; since every function is a relation, we also write $\delta\subseteq Q\times\Sigma\times Q$. The automaton is \emph{deterministic} (DFA) if $|I|=1$ and $|\delta(q,a)|\le1$ for all $q\in Q$ and $a\in\Sigma$. Its language is $L(G)=\{w : \delta(I,w)\cap F\neq\emptyset\}$.

	Every plant in this paper is prefix-closed, and we take $F=Q$, writing \emph{NFA with all states final}~\cite{KRS09}. The languages of such automata are the prefix-closed regular languages, and we omit $F$ from the tuple when it equals $Q$. Automata with $F\neq Q$ arise only as auxiliary constructions inside proofs.

	For alphabets $\Sigma_1\subseteq\Sigma$, the \emph{(natural) projection} $p\colon\Sigma^*\to\Sigma_1^*$ is the morphism with $p(a)=a$ for $a\in\Sigma_1$ and $p(a)=\varepsilon$ otherwise; we also write $s|_{\Sigma_1}$ for $p(s)$. For a morphism $f$ and a language $K$, the \emph{inverse image} is $f^{-1}(K)=\{s : f(s)\in K\}$. Languages $L_i\subseteq\Sigma_i^*$ are composed by the \emph{synchronous product}
	\(
		\|_{i=1}^{n} L_i = \{ s \in (\bigcup_i \Sigma_i)^* : s|_{\Sigma_i}\in L_i \text{ for all } i\}
	\);
	if the alphabets are pairwise disjoint, the product is called the \emph{shuffle}.

	Under partial observation, the alphabet is partitioned as $\Sigma=\So\mathbin{\dot\cup}\Suo$ into observable and unobservable events, and $P\colon\Sigma^*\to\So^*$ denotes the \emph{observation}.

\begin{definition}\label{def:relabeling}
	Let $T=\To\mathbin{\dot\cup}\Tuo$ be a \emph{template alphabet} disjoint from $\Sigma$. A \emph{relabeling} is a surjective map $R\colon\Sigma\to T$ that \emph{preserves the observability status of events}, that is, $R(\So)\subseteq\To$ and $R(\Suo)\subseteq\Tuo$, extended to a morphism $R\colon\Sigma^*\to T^*$.
\end{definition}

	A relabeling is \emph{letter-to-letter}: $|R(s)|=|s|$ for every $s$, and no event is erased. In general, $R$ is not injective; we say that $R$ \emph{merges} $a\neq b$ if $R(a)=R(b)$, and that $R$ is \emph{injective on} $\Sigma'\subseteq\Sigma$ if its restriction to $\Sigma'$ is injective. We write $\PT\colon T^*\to\To^*$ for the observation on the template side and $\Ro\colon\So^*\to\To^*$ for the morphism induced by the restriction of $R$ to $\So$. Since morphisms agreeing on letters coincide, preservation of the observability status ensures that the diagram of Fig.~\ref{fig:maps} commutes:
	\begin{equation}\label{eq:comm}
		\PT\circ R \;=\; \Ro\circ P \,.
	\end{equation}

  \begin{figure}
		\centering
		\[
		\begin{CD}
			\Sigma^* @>{R}>> T^* \\
			@V{P}VV @VV{\PT}V \\
			\So^* @>>{\Ro}> \To^*
		\end{CD}
		\]
		\caption{Relabelings act horizontally, observations vertically. The diagram commutes because $R$ preserves the observability status of events.}
		\label{fig:maps}
  \end{figure}

\begin{definition}[ROC~{\cite[Def.~1]{LKL22}}]\label{def:roc}
	A relabeling $R$ is \emph{relabeling observation consistent} (ROC) with respect to a prefix-closed language $L\subseteq\Sigma^*$ and the observation $P$ if for every $s\in L$ and every $t'\in R(L)$ with $\PT(R(s))=\PT(t')$ there exists $s'\in L$ with $R(s')=t'$ and $P(s')=P(s)$.
\end{definition}

	Every plant string $s$ must thus be matched, on the observation level, by a witness of every template string that the template-level observer cannot distinguish from $R(s)$.

\begin{problem}[ROC verification]\label{prob:roc}
	Given an NFA $G$ over $\Sigma$, a partition $\Sigma=\So\mathbin{\dot\cup}\Suo$, and a relabeling $R$, decide whether $R$ is ROC with respect to $L(G)$ and $P$.
\end{problem}

\section{Compatible and Realized Pairs}\label{sec:pairs}
	Definition~\ref{def:roc} quantifies over a plant string and a template string. It is convenient to replace the plant string by its observation, which is all that the definition uses.

	Define the set of \emph{realized} pairs and the set of \emph{compatible} pairs,
	\begin{align*}
		\Jm(L) &:= \{ (R(z),\, P(z)) : z\in L \}\ \subseteq\ T^*\times\So^*\,, \\
		\Fm(L) &:= \{ (t',x)\in R(L)\times P(L) : \PT(t')=\Ro(x) \}\,,
	\end{align*}
	and say that $z\in L$ \emph{realizes} the pair $(R(z),P(z))$. By~\eqref{eq:comm}, $\Jm(L)\subseteq\Fm(L)$: a realized pair is compatible. The set $\Fm(L)$ is the largest relation that $\Jm(L)$ could conceivably be, as it collects the pairs that the commuting diagram of Fig.~\ref{fig:maps} does not rule out.

\begin{lemma}\label{lem:pair}
	A relabeling $R$ is ROC with respect to a prefix-closed language $L$ and $P$ if and only if $\Jm(L)=\Fm(L)$, that is, if and only if every compatible pair is realized.
\end{lemma}
\begin{proof}
	The inclusion $\Jm(L)\subseteq\Fm(L)$ follows by~\eqref{eq:comm}. Assume that $R$ is ROC and let $(t',x)\in\Fm(L)$. Let $s\in L$ with $P(s)=x$. Then $\PT(R(s))=\Ro(P(s))=\Ro(x)=\PT(t')$, and ROC provides $s'\in L$ with $R(s')=t'$ and $P(s')=P(s)=x$; hence $(t',x)\in\Jm(L)$.

	Conversely, assume $\Jm(L)=\Fm(L)$ and let $s\in L$ and $t'\in R(L)$ satisfy $\PT(R(s))=\PT(t')$. Setting $x=P(s)\in P(L)$, we get $\Ro(x)=\PT(R(s))=\PT(t')$ by~\eqref{eq:comm}, and hence $(t',x)\in\Fm(L)=\Jm(L)$; this is precisely the existence of $s'\in L$ with $P(s')=x=P(s)$ and $R(s')=t'$.
\end{proof}

	We first determine when the condition is vacuous. The following characterization sharpens the sufficient conditions used in the multi-agent literature, showing that they are not merely sufficient but exactly the trivializing cases.

\begin{proposition}\label{prop:trivial-roc}
	Let $R$ be a relabeling. Then $R$ is ROC with respect to every prefix-closed language $L\subseteq\Sigma^*$ and $P$ if and only if $\Suo=\emptyset$ or $R$ is injective on $\So$. In particular, this holds when $\So=\emptyset$.
\end{proposition}
\begin{proof}
	Assume first $\Suo=\emptyset$. Then $\Tuo=\emptyset$ by surjectivity of $R$, so $\PT$ is the identity on $T^*$, and $R(z)=\Ro(P(z))$ for every $z$. Let $s\in L$ and $t'\in R(L)$ with $\PT(R(s))=\PT(t')$; this reads $R(s)=t'$, and $s'=s$ is a witness.

	Assume now that $R$ is injective on $\So$. Being letter-to-letter, $\Ro\colon\So^*\to\To^*$ is then injective. Let $s\in L$ and $t'\in R(L)$ with $\PT(R(s))=\PT(t')$, and choose $s'\in L$ with $R(s')=t'$. By~\eqref{eq:comm}, $\Ro(P(s'))=\PT(R(s'))=\PT(t')=\PT(R(s))=\Ro(P(s))$, and injectivity yields $P(s')=P(s)$. If $\So=\emptyset$, the restriction of $R$ to $\So$ is vacuously injective.

	Conversely, suppose $\Suo\neq\emptyset$ and that $R$ merges two observable events, say $R(a)=R(b) = \nu\in\To$ with $a\neq b$ in $\So$. Fix $u\in\Suo$ and put $\tau = R(u)\in\Tuo$. Consider the prefix-closed language $L=\{\varepsilon,a,u,ub\}$. Take $s=ub$ and $t'=\nu$. Then $t'=R(a)\in R(L)$ and $\PT(R(s))=\PT(\tau\nu)=\nu=\PT(t')$, so the pair is compatible. A witness $s'$ would satisfy $R(s')=\nu$ and $P(s')=P(ub)=b$; the only string of $L$ with $R$-image $\nu$ is $a$, and $P(a)=a\neq b$. Hence $R$ is not ROC with respect to $L$.
\end{proof}

	Injectivity of $R$ on $\So$ says that the observation is recoverable from the template; the relabeling may still merge unobservable events arbitrarily. In the multi-agent setting, ROC can thus fail only if the template identifies observable events of different agents, which has an immediate consequence for the hardness results below: every reduction must merge observable events.

\section{ROC Verification is \PSPACE-Complete}\label{sec:complete}

\subsection{The Upper Bound}\label{ssec:upper}
	Compatible pairs are pairs of strings of different lengths over different alphabets, and the first step is to encode such a pair as a single word. Put
	\[
		\Gamma \;:=\; \bigl(\Tuo\times\{\varepsilon\}\bigr) \;\mathbin{\dot\cup}\; \bigl\{(\tau,c)\in\To\times\So : R(c)=\tau\bigr\}\,,
	\]
	a finite alphabet of pairs, and let $\pi_1\colon\Gamma^*\to T^*$ and $\pi_2\colon\Gamma^*\to\So^*$ be the morphisms projecting onto the first and the second component, so that $\pi_1(\tau,\sigma)=\tau$ and $\pi_2(\tau,\sigma)=\sigma$. Write $\pair(w):=(\pi_1(w),\pi_2(w))$. Finally, let $h\colon\Sigma^*\to\Gamma^*$ be the morphism
	\[
		h(c) = \begin{cases}
			(R(c),\varepsilon) & \text{if } c\in\Suo,\\
			(R(c),c)           & \text{if } c\in\So .
		\end{cases}
	\]
	The map $h$ is well defined because $R$ preserves the observability status, and comparing on letters gives
	\begin{equation}\label{eq:hproj}
		\pi_1 h = R \qquad\text{and}\qquad \pi_2 h = P \,.
	\end{equation}

\begin{lemma}\label{lem:pairs}
	Let $Z = \{(t',x)\in T^*\times\So^* : \Ro(x)=\PT(t')\}$. Then $\pair$ is a bijection from $\Gamma^*$ onto $Z$.
\end{lemma}
\begin{proof}
	That $\pair(w)\in Z$ for every $w\in\Gamma^*$ follows by applying $\Ro$ letterwise: a letter of $\To\times\So$ contributes a letter to both components with matching $R$-image, and a letter of $\Tuo\times\{\varepsilon\}$ contributes a $\Tuo$-letter to the first component only, which $\PT$ erases.

	For surjectivity, let $(t',x)\in Z$ and write $t'=\tau_1\cdots\tau_\ell$. Let $j_1<\cdots<j_k$ enumerate the positions with $\tau_j\in\To$, so that $\PT(t')=\tau_{j_1}\cdots\tau_{j_k}$. Since $\Ro(x)=\PT(t')$ and $\Ro$ is letter-to-letter, $x$ has length $k$ and its $r$-th letter $\sigma_r$ satisfies $R(\sigma_r)=\tau_{j_r}$. The word $w=\gamma_1\cdots\gamma_\ell$ with $\gamma_{j_r}=(\tau_{j_r},\sigma_r)$ and $\gamma_j=(\tau_j,\varepsilon)$ for $j\notin\{j_1,\dots,j_k\}$ is a word over $\Gamma$ with $\pair(w)=(t',x)$.

	For injectivity, note that this $w$ is the only preimage. A letter of $w$ contributes a $\To$-letter to $\pi_1(w)$ exactly when it lies in $\To\times\So$; hence the positions $j_1,\dots,j_k$ are determined by $t'$, and so are the first components of all letters. The second components are then forced: those at the positions $j_r$ are the letters of $x$ in order, and the remaining ones are $\varepsilon$.
\end{proof}

	The bijection is what makes the relabeling case tractable. Every compatible pair has exactly one code, so a set of pairs is a language, and comparing two sets of pairs is comparing two languages.

\begin{theorem}\label{thm:upper}
	ROC verification is in \PSPACE. Specifically, with
	\[
		W = \pi_1^{-1}(R(L))\cap\pi_2^{-1}(P(L))\ \subseteq\ \Gamma^*\,,
	\]
	$R$ is ROC with respect to $L$ and $P$ if and only if $W\subseteq h(L)$, and NFAs $\mathcal{A}_W$ and $\mathcal{A}_h$ for $W$ and $h(L)$ with $O(n^2)$ and $n$ states, respectively, are computable in polynomial time from an $n$-state NFA $G$ for $L$.
\end{theorem}
\begin{proof}
	By definition, $w\in W$ if and only if $\pi_1(w)\in R(L)$ and $\pi_2(w)\in P(L)$, so $\pair(W)=Z\cap(R(L)\times P(L))=\Fm(L)$. By Lemma~\ref{lem:pairs}, $\pair$ restricts to a bijection between $W$ and the compatible pairs.

	We show that $R$ is ROC if and only if $W\subseteq h(L)$. By~\eqref{eq:hproj}, $\pair(h(s))=(R(s),P(s))$ for every $s\in\Sigma^*$; as $\pair$ is injective, $w=h(s)$ if and only if $\pair(w)=(R(s),P(s))$. Assume ROC and let $w\in W$. Then $\pair(w)=(t',x)$ is a compatible pair, and Lemma~\ref{lem:pair} yields $s'\in L$ with $R(s')=t'$ and $P(s')=x$. Then $\pair(h(s'))=(t',x)=\pair(w)$, and hence $w=h(s')\in h(L)$. Conversely, assume $W\subseteq h(L)$ and let $(t',x)$ be compatible. Its preimage $w=\pair^{-1}(t',x)$ belongs to $W$, and hence $w=h(s')$ for some $s'\in L$, and~\eqref{eq:hproj} gives $R(s')=\pi_1(w)=t'$ and $P(s')=\pi_2(w)=x$; the pair is realized and Lemma~\ref{lem:pair} applies.

	Regular languages are closed under morphisms, inverse morphisms, and intersection, thus $W$ and $h(L)$ are regular. An NFA for $R(L)$ is obtained from $G$ by relabeling its transitions by $R$, and one for $P(L)$ by relabeling them by $P$ and removing the resulting $\varepsilon$-transitions; taking inverse morphic images adds no states, so $\mathcal{A}_W$ has $O(n^2)$ states after the product, and $\mathcal{A}_h$, obtained by relabeling the transitions of $G$ by $h$, has $n$ states. All states of both automata are final. Deciding inclusion of NFAs is in \PSPACE~\cite{MS72},~\cite[Ch.~10]{HMU06}.
\end{proof}

	We write $\mathcal{A}_W$ and $\mathcal{A}_h$ for these two automata throughout. Since $h$ is letter-to-letter, $\mathcal{A}_h$ is literally $G$ with relabeled transitions.

\subsection{Hardness}\label{ssec:hard}
	The reductions are from the universality problem for NFAs with all states final, which is \PSPACE-complete even over a binary alphabet~\cite{KRS09}: given an NFA $A$ over $\Delta$, \emph{universality} asks whether $L(A)=\Delta^*$. We use the following normal form.

\begin{lemma}\label{lem:pad}
	The universality problem for NFAs with all states final over $\Delta=\{a,b\}$ is \PSPACE-complete under the restriction that the input automaton $A$ has a unique initial state with no incoming transitions and satisfies $a^*\subseteq L(A)$.
\end{lemma}
\begin{proof}
	Given an NFA $A_0$ over $\Delta$ with all states final, let $A$ consist of a fresh initial state $q_0$, a state $u$ looping under $a$ and $b$, the transition $(q_0,a,u)$, a copy of $A_0$, and the transitions $(q_0,b,i)$ for every initial state $i$ of $A_0$. Then $A$ has all states final, $q_0$ has no incoming transitions, and $L(A)=\{\varepsilon\}\cup a\Delta^*\cup b\cdot L(A_0)$ is prefix-closed, contains $a^*$, and is universal if and only if $L(A_0)$ is.
\end{proof}

	The first hardness result concerns nondeterministic plants and uses a single unobservable event.

\begin{lemma}\label{lem:nfa-hard}
	ROC verification is \PSPACE-hard for NFAs, even if $|\So|=2$ and $|\Suo|=|\To|=|\Tuo|=1$, that is, if $R$ merges only one pair of observable events.
\end{lemma}
\begin{proof}
	Let $A$ over $\{a,b\}$ be as in Lemma~\ref{lem:pad}, with initial state $q_0$. Define $\So=\{a,b\}$ and $\Suo=\{u\}$, and the relabeling $R$ onto $\To=\{\hat a\}$ and $\Tuo=\{\hat u\}$ by $R(a)=R(b)=\hat a$ and $R(u)=\hat u$. The prefix-closed language
	\(
	  	L=\{a,b\}^*\cup u\cdot L(A)
	\)
	is recognized by an NFA with all states final and two states more than $A$: a fresh initial state $p_0$, a state $r$ with self-loops under $a$ and $b$, and the transitions $(p_0,a,r)$, $(p_0,b,r)$, and $(p_0,u,q_0)$.

	Since $a^*\subseteq L(A)$, we have $R(L)=\hat a^*\cup\hat u\,\hat a^*$ and $P(L)=\{a,b\}^*\cup L(A)=\{a,b\}^*$. For $t'\in R(L)$ and $x\in P(L)$, both $\PT(t')$ and $\Ro(x)$ are powers of $\hat a$, and compatibility of $(t',x)$ says that the number of $\hat a$'s in $t'$ equals $|x|$. The compatible pairs are therefore $(\hat a^{|w|},w)$ and $(\hat u\,\hat a^{|w|},w)$ for $w\in\{a,b\}^*$.

	A witness for $(\hat a^{|w|},w)$ contains no unobservable event and has observation $w$; hence it is $w$ itself, and $w\in\{a,b\}^*\subseteq L$, so the pair is realized. A witness for $(\hat u\,\hat a^{|w|},w)$ is of the form $uv$ with $v$ observable and $P(uv)=v=w$; thus it is $uw$, and $uw\in L$ if and only if $w\in L(A)$. Hence $R$ is ROC with respect to $L$ and $P$ if and only if $L(A)=\{a,b\}^*$, which is \PSPACE-hard to decide by Lemma~\ref{lem:pad}.
\end{proof}

	To show hardness for DFAs, we need more unobservable events, since the situation of Lemma~\ref{lem:nfa-hard} is tractable for DFAs by Corollary~\ref{cor:single-uo} below.

\begin{lemma}\label{lem:dfa-hard}
	ROC verification is \PSPACE-hard for DFAs, even if $|\So|=|\Suo|=2$ and $|\To|=|\Tuo|=1$, that is, if $R$ merges one pair of observable and one pair of unobservable events.
\end{lemma}
\begin{proof}
	Let $A$ over $\{a,b\}$ be as in Lemma~\ref{lem:pad}, with initial state $q_0$ and state set $Q$; enumerate its transitions as $t_i=(p_i,c_i,q_i)$ for $1\le i\le k$, and let $\ell=\lceil\log_2 k\rceil$, where $\ell\ge1$ since $k\ge3$ by the construction of Lemma~\ref{lem:pad}. Define $\So=\{a,b\}$ and $\Suo=\{u_0,u_1\}$, and let $R(a)=R(b)=\hat a$ and $R(u_0)=R(u_1)=\hat u$. The construction determinizes the choice among the transitions of $A$: before each observable step, a block of $\ell$ unobservable events selects the transition to apply, and the two selectors are indistinguishable on the template because $R$ merges them.

	The DFA $G$, with all states final and initial state $p_0$, consists of two parts. In the \emph{loop branch}, the letters $a,b$ lead from $p_0$ to a state looping under $\{a,b\}$; this part generates $\{a,b\}^*$. In the \emph{run branch}, there is, for every $q\in Q$, a complete binary tree with nodes $(q,\beta)$ for $\beta\in\{0,1\}^{\le\ell}$ and the transitions $\bigl((q,\beta),u_c,(q,\beta c)\bigr)$; identifying the leaves with the indices $1,\dots,2^{\ell}$, the leaf with index $i$ carries the single transition under $c_i$ to $(q_i,\varepsilon)$ if $i\le k$ and $p_i=q$. Finally, add the transitions $\bigl(p_0,u_c,(q_0,c)\bigr)$ for $c\in\{0,1\}$. The automaton is deterministic of size $O(k\cdot|Q|)$; the run branch is reentered only at the tree roots, which carry only unobservable transitions, so the two parts do not mix. Let $L=L(G)$.

	A run-branch string alternates blocks of $\ell$ unobservable letters with single observable letters, possibly ending inside a block; conversely, every such pattern is attained, since $a^*\subseteq L(A)$ and the trees are complete. Hence $P(L)=\{a,b\}^*$ and
	\(
		R(L)=\hat a^*\cup\{(\hat u^{\ell}\hat a)^n\,\hat u^{\,j} : n\ge0,\ 0\le j\le\ell\}.
	\)
	As in Lemma~\ref{lem:nfa-hard}, a pair $(t',x)$ is compatible if and only if the number of $\hat a$'s in $t'$ equals $|x|$. A pair $(\hat a^{|w|},w)$ is realized by the loop-branch string $w$ itself. The remaining pairs are $\bigl((\hat u^{\ell}\hat a)^n\,\hat u^{\,j},w\bigr)$ with $w=w_1\cdots w_n$ and $n+j\ge1$; the template and the observation force a witness to be of the form $\varphi_1w_1\cdots\varphi_nw_n\,\psi$ with $\varphi_m\in\{u_0,u_1\}^{\ell}$ and $\psi\in\{u_0,u_1\}^{j}$, which contains an unobservable letter and hence lies in the run branch. There, each block $\varphi_m$ descends the tree of the current state, and $w_m$ is enabled if $\varphi_m$ selects a transition of $A$ under $w_m$ from that state; the trailing $\psi$ is always readable. A witness thus exists if and only if $A$ has a run on $w$, that is, if and only if $w\in L(A)$. Hence $R$ is ROC with respect to $L$ and $P$ if and only if $L(A)=\{a,b\}^*$.
\end{proof}

\begin{theorem}\label{thm:complete}
	ROC verification is \PSPACE-complete for both NFAs and DFAs.
\end{theorem}
\begin{proof}
	Membership is Theorem~\ref{thm:upper}, hardness for NFAs is Lemma~\ref{lem:nfa-hard}, and hardness for DFAs is Lemma~\ref{lem:dfa-hard}.
\end{proof}

	Theorem~\ref{thm:complete} rules out polynomial-time verification unless $\PTIME=\PSPACE$, and randomized polynomial-time verification unless $\PSPACE\subseteq\BPP$~\cite{AB09}.

\subsection{A Polynomial-Time Fragment}\label{ssec:ptime}

\begin{proposition}\label{prop:ptime}
	ROC verification is decidable in polynomial time for DFAs with a relabeling that is injective on unobservable events.
\end{proposition}
\begin{proof}
	We use the notation of Theorem~\ref{thm:upper}. The letters $h(c)$, for $c\in\Sigma$, are pairwise distinct: for $c\in\So$ they differ in the second component, for $c\in\Suo$ they differ in the first by the injectivity of $R$ on $\Suo$, and the two kinds differ as $\To\cap\Tuo=\emptyset$. Hence $h$ is injective, and $\mathcal{A}_h$, obtained by relabeling the transitions of the DFA $G$ by $h$, is deterministic; completing it with a sink state and making the sink the only final state gives a DFA for the complement of $h(L)$ with $|Q|+1$ states. Since $\mathcal{A}_W$ is an NFA with $O(|Q|^2)$ states and $R$ fails ROC if and only if $W$ intersects the complement of $h(L)$, a product construction and reachability decide ROC in polynomial time.
\end{proof}

	Since every relabeling on an unobservable alphabet with at most one event is injective, we obtain the following.

\begin{corollary}\label{cor:single-uo}
	For DFAs with $|\Suo|\le1$, ROC verification is decidable in polynomial time. \QEDbox
\end{corollary}

	Proposition~\ref{prop:trivial-roc} and Corollary~\ref{cor:single-uo} show that the parameters of Lemmata~\ref{lem:nfa-hard} and~\ref{lem:dfa-hard} cannot be improved: without merged observable events every instance is positive, and for DFAs the merging of unobservable events is unavoidable unless $\PTIME=\PSPACE$. Table~\ref{tab:landscape} summarizes.

\begin{table}
  \caption{ROC verification and unobservable events.}
  \label{tab:landscape}
  \centering
  \begin{tabular}{lll}
    \toprule
    Model & $R$ injective on $\Suo$ & $R$ merges events of $\Suo$ \\
    \midrule
    DFA & \PTIME (Prop.~\ref{prop:ptime})
              & \PSPACE-complete (Lem.~\ref{lem:dfa-hard}) \\
    NFA & \PSPACE-complete (Lem.~\ref{lem:nfa-hard})
              & \PSPACE-complete (Lem.~\ref{lem:dfa-hard}) \\
    \bottomrule
  \end{tabular}
\end{table}

\section{Sufficient Conditions}\label{sec:suff}
	By Theorem~\ref{thm:complete}, ROC verification is intractable in the worst case; by Proposition~\ref{prop:trivial-roc}, a condition on $(\Sigma,R)$ alone implies ROC for every language only if it implies that $\Suo=\emptyset$ or that $R$ is injective on $\So$. This section collects conditions that take the plant into account.

\subsection{Saturation}\label{ssec:sat}
	The following condition formalizes the interchangeability of the instantiations of an observable template event: wherever one of them can occur in the plant, every other can.

\begin{definition}\label{def:sat}
	A language $L\subseteq\Sigma^*$ is \emph{$\Ro$-saturated} if for every string $uav\in L$ with $a\in\So$ and every $b\in\So$ with $R(b)=R(a)$, also $ubv\in L$.
\end{definition}

\begin{proposition}\label{prop:sat}
	If a prefix-closed language $L$ is $\Ro$-saturated, then $R$ is ROC with respect to $L$ and $P$.
\end{proposition}
\begin{proof}
	Let $(t',x)$ be a compatible pair. Since $t'\in R(L)$, there is $s_0\in L$ with $R(s_0)=t'$. By~\eqref{eq:comm}, $\Ro(P(s_0))=\PT(t')=\Ro(x)$; since $\Ro$ is letter-to-letter, $P(s_0)$ and $x$ have the same length and, positionwise, the $j$-th letter of $x$ has the same $R$-image as the $j$-th letter of $P(s_0)$. Replace, in $s_0$, the observable letters by the corresponding letters of $x$ one at a time; each intermediate string lies in $L$ by saturation, and the final string $s'$ satisfies $P(s')=x$ and, since each replacement is by an event with the same $R$-image, $R(s')=R(s_0)=t'$. Hence the pair is realized, and Lemma~\ref{lem:pair} applies.
\end{proof}

	Proposition~\ref{prop:sat} subsumes the injectivity case of Proposition~\ref{prop:trivial-roc}: if $R$ is injective on $\So$, then no two observable events share an $R$-image, and hence every language is $\Ro$-saturated.

	Saturation of $L(G)$ is a language property, decidable by the inclusion $L(G^{\mathrm{sat}})\subseteq L(G)$, where $G^{\mathrm{sat}}$ arises from $G$ by adding, for every transition $(q,a,q')$ with $a\in\So$ and every $b$ with $R(b)=R(a)$, the transition $(q,b,q')$. Indeed, if the inclusion holds and $uav\in L(G)$ with $R(b)=R(a)$, then the run of $uav$ traverses some $(q,a,q')$ and the added transition $(q,b,q')$ gives $ubv\in L(G^{\mathrm{sat}})\subseteq L(G)$; conversely, if $L(G)$ is saturated, then a word of $L(G^{\mathrm{sat}})$ is turned into a word of $L(G)$ by replacing its added transitions by the original ones one at a time and applying saturation in the reverse direction at each step. The inclusion is polynomial-time decidable if $G$ is a DFA, and decidable in polynomial space in general.

	Saturation makes no disjointness assumption and hence applies to agents that share events, but it is only sufficient: for two agents with behaviors $\overline{a_1c_1}$ and $\overline{a_2c_2}$, where $R(a_1)=R(a_2)$ is observable and $R(c_1)=R(c_2)$ is unobservable, the shuffle contains $a_1a_2$ but not $a_1a_1$, and hence is not saturated, yet it is ROC, as one checks on the finitely many compatible pairs.

\subsection{Automaton-Level Tests}\label{ssec:aut}
	A \emph{simulation} of an NFA $\mathcal{A}_1$ by an NFA $\mathcal{A}_2$ is a relation $S$ between their state sets such that whenever $(p,q)\in S$ and $(p,c,p')$ is a transition of $\mathcal{A}_1$, there is a transition $(q,c,q')$ of $\mathcal{A}_2$ with $(p',q')\in S$, and such that $q$ is final in $\mathcal{A}_2$ whenever $(p,q)\in S$ and $p$ is final in $\mathcal{A}_1$; the second requirement is vacuous when all states of $\mathcal{A}_2$ are final. The simulation \emph{covers the initial states} if every initial state of $\mathcal{A}_1$ is $S$-related to an initial state of $\mathcal{A}_2$, in which case $L(\mathcal{A}_1)\subseteq L(\mathcal{A}_2)$. A \emph{bisimulation} on an NFA is a symmetric simulation of the NFA by itself, and two states are \emph{bisimilar} if some bisimulation relates them. The largest simulation and the largest bisimulation are computable in polynomial time~\cite{HHK95,PT87}.

\begin{lemma}\label{lem:sat-aut}
	Let $G$ be an NFA with all states final. Then (i) implies (ii), and (ii) implies that $L(G)$ is $\Ro$-saturated:
\begin{enumerate}
  \item[(i)] \emph{(state saturation)} for every transition $(q,a,q')$ of $G$ with $a\in\So$ and every $b\in\So$ with $R(b)=R(a)$, also $(q,b,q')$ is a transition of $G$;
  \item[(ii)] \emph{(saturation up to bisimulation)} for every transition $(q,a,q')$ of $G$ with $a\in\So$ and every $b\in\So$ with $R(b)=R(a)$, there is a transition $(q,b,q'')$ of $G$ with $q''$ bisimilar to $q'$.
\end{enumerate}
\end{lemma}
\begin{proof}
	(i) implies (ii) with the identity bisimulation. For (ii), let $uav\in L(G)$ along a path reaching $q$ after $u$, taking $(q,a,q')$, and reading $v$ from $q'$. By (ii), there is $(q,b,q'')$ with $q''$ bisimilar to $q'$; bisimilar states of an all-states-final NFA enable the same strings, and hence $v$ is readable from $q''$ and $ubv\in L(G)$.
\end{proof}

	Condition (i) is verifiable in time $O(|\delta|\cdot|\So|)$ by scanning the transitions of $G$, and condition (ii) in polynomial time by additionally computing the largest bisimulation.

	Theorem~\ref{thm:upper} characterizes ROC as an inclusion of two automata, and replacing the inclusion by the existence of a simulation gives a further polynomial-time test.

\begin{proposition}\label{prop:sim}
	If some simulation of $\mathcal{A}_W$ by $\mathcal{A}_h$ covers the initial states, then $R$ is ROC with respect to $L$ and $P$. The hypothesis is verifiable in time polynomial in the size of $G$, and it is not necessary unless $\PTIME=\PSPACE$.
\end{proposition}
\begin{proof}
	A simulation covering the initial states implies $L(\mathcal{A}_W)\subseteq L(\mathcal{A}_h)$, that is, $W\subseteq h(L)$, which is ROC by Theorem~\ref{thm:upper}. The largest simulation of one NFA by another is computable in polynomial time in their sizes~\cite{HHK95}, and the covering of the initial states is then checked directly; both automata are polynomial in $|G|$ by Theorem~\ref{thm:upper}. Were the hypothesis also necessary, ROC verification would be in \PTIME, contradicting Theorem~\ref{thm:complete}.
\end{proof}

	Saturation and the simulation test are independent. The simulation test does not imply saturation: for $L=\overline{\{b\}}$ with $\So=\{a,b\}$, $R(a)=R(b)=\hat a$, and any nonempty $\Suo$, the only compatible pairs are $(\varepsilon,\varepsilon)$ and $(\hat a,b)$, so $W=h(L)$; on the two-state DFA for $L$, the reachable states of $\mathcal{A}_W$ are the pairs $(q,q)$, relating $(q,q)$ to $q$ is a simulation covering the initial states, and $L$ is not $\Ro$-saturated because $a\notin L$. Nor does saturation imply the simulation test, which, unlike saturation, depends on the automaton presenting the language. Let $G$ have the states $0$, $1$, $2$, the initial state $0$, and the transitions $(0,a,1)$, $(0,a,2)$, $(0,b,1)$, $(0,b,2)$, $(1,b,0)$, $(1,u,0)$, and $(2,a,0)$, where $\So=\{a,b\}$, $\Suo=\{u\}$, $R(a)=R(b)=\hat a$, and $R(u)=\hat u$. Then $L(G)$ consists of the strings whose odd positions carry observable events; membership depends only on the positions of $u$, so $L(G)$ is $\Ro$-saturated, and $W\subseteq h(L(G))$ holds by Proposition~\ref{prop:sat} and Theorem~\ref{thm:upper}. Yet no simulation of $\mathcal{A}_W$ by $\mathcal{A}_h$ covers the initial states: the initial state $(0,0)$ of $\mathcal{A}_W$ would have to be related to the initial state $0$ of $\mathcal{A}_h$, and its $(\hat a,a)$-successor $(2,1)$ to $1$ or to $2$; but $(2,1)$ enables both $(\hat a,b)$, its second component reading $b$ directly, and $(\hat a,a)$, its second component reading $a$ through the $\varepsilon$-transition left by erasing $u$, whereas $1$ does not enable $(\hat a,a)$ and $2$ does not enable $(\hat a,b)$. The two tests compare different objects: saturation compares the events identified by $R$, the simulation two derived automata whose branching reflects that of $G$.

\section{Compositionality}\label{sec:comp}
	Relabelings $R_i\colon\Sigma_i\to T_i$, for $i=1,\dots,n$, \emph{agree on shared events} if $R_i|_{\Sigma_i\cap\Sigma_j}=R_j|_{\Sigma_i\cap\Sigma_j}$ for all $i\ne j$, in which case the union $R=\bigcup_i R_i$ is a well-defined map on $\bigcup_i\Sigma_i$. They have \emph{pairwise disjoint templates} if they agree on shared events and, in addition, identify no events \emph{across} the components, that is, if for all $i\ne j$,
	\[
		R_i(c_i)=R_j(c_j) \text{ with } c_i\in\Sigma_i \text{ and } c_j\in\Sigma_j \implies c_i=c_j \,,
	\]
	in which case $c_i=c_j\in\Sigma_i\cap\Sigma_j$. Two events of the \emph{same} component may still share a template image, which is what a group of isomorphic agents requires. For pairwise disjoint alphabets, the agreement is trivial and the second condition reduces, by surjectivity of $R_i$, to $T_i\cap T_j=\emptyset$.

	Disjoint templates localize the template image of an event: if $c\in\Sigma_j$ satisfies $R(c)\in T_i$ for some $i\ne j$, then, since $R_i$ is surjective onto $T_i$, there is $c_i\in\Sigma_i$ with $R(c_i)=R(c)$, and hence $c=c_i\in\Sigma_i\cap\Sigma_j$; together with the trivial converse, for every $c$ and every $i$,
	\begin{equation}\label{eq:type}
		R(c)\in T_i \iff c\in\Sigma_i \,.
	\end{equation}
	Condition~\eqref{eq:type} is the requirement that $R$ ``preserves the local status of events'' imposed in~\cite{LKL22}. Note, however, that disjointness of templates is \emph{strictly stronger}, since~\eqref{eq:type} alone does not forbid two distinct events of different components from having the same image, and that the difference is not cosmetic: Proposition~\ref{prop:comp} fails under~\eqref{eq:type} alone, as the counterexample following its proof shows.

	Throughout, $\Sigma_i=\Sigma_{i,o}\mathbin{\dot\cup}\Sigma_{i,uo}$ and $T_i=T_{i,o}\mathbin{\dot\cup}T_{i,uo}$, and shared events have the same observability status in all components containing them, so that $\So=\bigcup_i\Sigma_{i,o}$ and $\Suo=\bigcup_i\Sigma_{i,uo}$ are well defined. We write $P_i$, $R_{i,o}$, and $P_{T_i}$ for the corresponding local maps.

\begin{lemma}\label{lem:amalg}
	Let $R_i\colon\Sigma_i\to T_i$, for $i=1,\dots,n$, be relabelings with pairwise disjoint templates, let $R=\bigcup_i R_i$, and let $t'\in(\bigcup_i T_i)^*$ with $t'_i=t'|_{T_i}$.
	\begin{enumerate}
		\item If $s_i\in\Sigma_i^*$ satisfy $R_i(s_i)=t'_i$ for every $i$, then there is $s\in(\bigcup_i\Sigma_i)^*$ with $R(s)=t'$ and $s|_{\Sigma_i}=s_i$ for every $i$.
		\item If $s,\tilde s\in(\bigcup_i\Sigma_i)^*$ satisfy $R(s)=R(\tilde s)=t'$ and $s|_{\Sigma_i}=\tilde s|_{\Sigma_i}$ for every $i$, then $s=\tilde s$.
	\end{enumerate}
\end{lemma}
\begin{proof}
	Let $\ell=|t'|$ and write $t'=\tau_1\cdots\tau_\ell$. For every $i$, let $J_i=\{j : \tau_j\in T_i\}$, enumerated as $j_{i,1}<\cdots<j_{i,m_i}$; then $\bigcup_i J_i=\{1,\dots,\ell\}$ and $t'_i=\tau_{j_{i,1}}\cdots\tau_{j_{i,m_i}}$.

	1) Since $R_i$ is letter-to-letter, $|s_i|=m_i$; write $s_i=c_{i,1}\cdots c_{i,m_i}$, so that $R(c_{i,k})=\tau_{j_{i,k}}$. Define $s=\sigma_1\cdots\sigma_\ell$ by $\sigma_{j_{i,k}}=c_{i,k}$, which assigns a letter to every position because $\bigcup_i J_i=\{1,\dots,\ell\}$. The assignment is consistent: if $j=j_{i,k}=j_{i',k'}$ with $i\ne i'$, then $R(c_{i,k})=\tau_j=R(c_{i',k'})$ with $c_{i,k}\in\Sigma_i$ and $c_{i',k'}\in\Sigma_{i'}$, and disjointness of templates gives $c_{i,k}=c_{i',k'}$. By construction $R(s)=t'$. Finally, by~\eqref{eq:type}, $\sigma_j\in\Sigma_i$ if and only if $j\in J_i$; hence $s|_{\Sigma_i}=c_{i,1}\cdots c_{i,m_i}=s_i$.

	2) Let $R(s)=t'$. Then $|s|=\ell$ and, by~\eqref{eq:type}, the positions of $s$ carrying letters of $\Sigma_i$ are exactly those of $J_i$; consequently the $j_{i,k}$-th letter of $s$ is the $k$-th letter of $s|_{\Sigma_i}$. The same applies to $\tilde s$, with the same sets $J_i$. Let $j\in\{1,\dots,\ell\}$; then $j=j_{i,k}$ for some $i$ and $k$, and the $j$-th letters of $s$ and $\tilde s$ are the $k$-th letters of $s|_{\Sigma_i}$ and $\tilde s|_{\Sigma_i}$, which coincide. Therefore $s=\tilde s$.
\end{proof}

\begin{proposition}\label{prop:comp}
	For $i=1,\dots,n$, let $L_i\subseteq\Sigma_i^*$ be nonempty and prefix-closed, and let $R_i\colon\Sigma_i\to T_i$ be surjective relabelings preserving the observability status and having pairwise disjoint templates. Let $L=\|_{i=1}^n L_i$ and $R=\bigcup_i R_i$.
	\begin{enumerate}
		\item If $R_i$ is ROC with respect to $L_i$ and $P_i$ for every $i$, then $R$ is ROC with respect to $L$ and $P$.
		\item If the alphabets $\Sigma_1,\dots,\Sigma_n$ are pairwise disjoint, then $R$ is ROC with respect to $L$ and $P$ if and only if $R_i$ is ROC with respect to $L_i$ and $P_i$ for every $i$.
	\end{enumerate}
\end{proposition}
\begin{proof}
	Write $\Sigma=\bigcup_i\Sigma_i$ and $T=\bigcup_i T_i$. Since $R$ agrees with each $R_i$ on $\Sigma_i$, it is a well-defined relabeling, and for $s\in\Sigma^*$ and every $i$,
	\begin{equation}\label{eq:comm-restr}
		R(s)|_{T_i}=R_i(s|_{\Sigma_i}) \quad\text{and}\quad P(s)|_{\Sigma_{i,o}}=P_i(s|_{\Sigma_i}) \,,
	\end{equation}
	because a letter $c$ of $s$ contributes to $R(s)|_{T_i}$ if and only if $R(c)\in T_i$, which is if and only if $c\in\Sigma_i$ by~\eqref{eq:type}; analogously for the observation. Moreover, the observable restrictions $R_{1,o},\dots,R_{n,o}$ inherit the hypotheses: they agree on $\Sigma_{i,o}\cap\Sigma_{j,o}$ because $R_i$ and $R_j$ agree on $\Sigma_i\cap\Sigma_j$, and an identification $\Ro(c_i)=\Ro(c_j)$ with $c_i\in\Sigma_{i,o}$, $c_j\in\Sigma_{j,o}$, and $i\ne j$ is an identification under $R$ and therefore $c_i=c_j$; also $\Ro=\bigcup_i R_{i,o}$. Consequently~\eqref{eq:type} and~\eqref{eq:comm-restr} hold at the observable level as well; in particular $T_i\cap\To=T_{i,o}$, so that $\PT$ agrees with $P_{T_i}$ on $T_i^*$, and
	\begin{equation}\label{eq:comm-restr-o}
		\Ro(x)|_{T_{i,o}}=R_{i,o}(x|_{\Sigma_{i,o}}) ~\text{ and }~ \PT(t')|_{T_{i,o}}=P_{T_i}(t'|_{T_i})
	\end{equation}
	for $x\in\So^*$ and $t'\in T^*$. We use Lemma~\ref{lem:pair} in both parts.

	1) Let $(t',x)$ be a compatible pair for $L$ and $R$. Put $t'_i=t'|_{T_i}$ and $x_i=x|_{\Sigma_{i,o}}$. From $t'=R(z)$ with $z\in L$, \eqref{eq:comm-restr} gives $t'_i=R_i(z|_{\Sigma_i})\in R_i(L_i)$, and from $x=P(y)$ with $y\in L$ it gives $x_i=P_i(y|_{\Sigma_i})\in P_i(L_i)$; restricting $\Ro(x)=\PT(t')$ to $T_{i,o}$ gives $R_{i,o}(x_i)=P_{T_i}(t'_i)$ by~\eqref{eq:comm-restr-o}. Therefore $(t'_i,x_i)$ is a compatible pair for $L_i$ and $R_i$, and by assumption there are witnesses $s'_i\in L_i$ with $R_i(s'_i)=t'_i$ and $P_i(s'_i)=x_i$.

	By Lemma~\ref{lem:amalg}(1), applied to $t'$ and the witnesses, there is $s'\in\Sigma^*$ with $R(s')=t'$ and $s'|_{\Sigma_i}=s'_i\in L_i$ for every $i$; hence $s'\in L$. It remains to show that $P(s')=x$. Apply Lemma~\ref{lem:amalg}(2) to the observable restrictions, with $\PT(t')\in\To^*$ in the role of the template string and with $P(s')$ and $x$ in the role of $s$ and $\tilde s$. Its hypotheses hold: the two strings have the same $\Ro$-image $\PT(t')$, since $\Ro(P(s'))=\PT(R(s'))=\PT(t')$ by~\eqref{eq:comm} and $\Ro(x)=\PT(t')$ by compatibility; and they have the same restrictions, since $P(s')|_{\Sigma_{i,o}}=P_i(s'_i)=x_i=x|_{\Sigma_{i,o}}$. The lemma yields $P(s')=x$.

	2) Fix $i$ and let $(t'_i,x_i)$ be a compatible pair for $L_i$ and $R_i$. Every $L_j$ is nonempty and prefix-closed, so $\varepsilon\in L_j$, and disjointness of the alphabets gives $z|_{\Sigma_j}=\varepsilon\in L_j$ for every $z\in L_i$ and $j\ne i$; hence $L_i\subseteq L$, and consequently $t'_i\in R(L)$ and $x_i\in P(L)$, whereas $\Ro(x_i)=R_{i,o}(x_i)=P_{T_i}(t'_i)=\PT(t'_i)$ shows that $(t'_i,x_i)$ is compatible for $L$ and $R$. By ROC of $R$ there is $s'\in L$ with $R(s')=t'_i$ and $P(s')=x_i$. Now $t'_i\in T_i^*$, so every letter $c$ of $s'$ satisfies $R(c)\in T_i$ and hence lies in $\Sigma_i$ by~\eqref{eq:type}; thus $s'\in\Sigma_i^*$ and $s'=s'|_{\Sigma_i}\in L_i$, and $P_i(s')=P(s')=x_i$. Together with part 1, the equivalence follows.
\end{proof}

	The hypotheses are tight in both parts. Template disjointness cannot be dropped from part 1, even for disjoint alphabets: Example~\ref{ex:sef} below is an instance of two components with disjoint alphabets, each ROC, whose shuffle is not ROC because the private events $a_1,a_2$ share the template image $\hat a$.

	Nor can template disjointness be weakened to the local-status condition~\eqref{eq:type}, although the two differ only slightly. Condition~\eqref{eq:type} holds if and only if any two events with a common template image belong to the same components, since $R(c)=R(c')$ gives $c\in\Sigma_i\iff R(c)\in T_i\iff R(c')\in T_i\iff c'\in\Sigma_i$ for every $i$; template disjointness requires in addition that the two events be equal. The weakening therefore permits two \emph{distinct shared} events to carry a common template image, and that alone destroys part 1. Take $\Sigma_1=\Sigma_2=\{a,b,u_1,u_2\}$ with $\So=\{a,b\}$, and let $R(a)=R(b)=\hat a$ and $R(u_1)=R(u_2)=\hat u$; condition~\eqref{eq:type} holds trivially, whereas $a$ and $b$ violate template disjointness. Put
	\[
		L_1=\overline{\{u_1a,\,u_1b,\,a,\,b\}},
		\qquad
		L_2=\overline{\{u_1a,\,u_2b,\,a,\,b\}} \,.
	\]
	Both are ROC: in $L_1$ the compatible pair $(\hat u\hat a,b)$ is realized by $u_1b$, and in $L_2$ by $u_2b$; the remaining pairs are realized by $a$, $b$, $u_1$, and $u_1a$ in either language. Their composition is their intersection, $L=\overline{\{u_1a,\,a,\,b\}}$, in which $(\hat u\hat a,b)$ is still compatible---$\hat u\hat a\in R(L)$ via $u_1a$ and $b\in P(L)$ via $b$---but no longer realized, since a witness would have to be $u_1b$ or $u_2b$ and the two components have kept different ones. The failure is exactly the one that Lemma~\ref{lem:amalg} rules out: the components choose different preimages of the same unobservable template event, and no string of the composition restricts to both choices.

	Alphabet disjointness cannot be dropped from part 2: with shared events, the composition may prune violating behaviors of a component. Let $L_1=\overline{sa}\cup\overline{sub}$ over $\Sigma_1=\{s,a,b,u\}$ with $u$ unobservable, $R_1(s)=\hat s$, $R_1(a)=R_1(b)=\hat a$, and $R_1(u)=\hat u$. The compatible pair $(\hat s\hat a,\,sb)$, with $\hat s\hat a\in R_1(L_1)$ via $sa$ and $sb\in P_1(L_1)$ via $sub$, is not realized, since a witness must be the observable string $sb\notin L_1$; hence $R_1$ is not ROC with respect to $L_1$. Take $L_2=\{\varepsilon\}$ over $\Sigma_2=\{s\}$ with $R_2(s)=\hat s$. The templates are disjoint, yet component~2 blocks the shared event $s$, so $L=L_1\|L_2=\{\varepsilon\}$ and $R=R_1\cup R_2$ is trivially ROC.

	In the framework of Liu \emph{et al.}~\cite{LCL19}, where the alphabets and the templates of different groups are disjoint, ROC of the multi-agent plant is therefore \emph{equivalent} to ROC of the individual groups. In the framework of Liu \emph{et al.}~\cite{LKL22}, where different groups may share events, part 1 reduces the \emph{confirmation} of ROC to the groups provided that the group relabelings have pairwise disjoint templates---the local-status condition~\eqref{eq:type} alone does not suffice, as the counterexample above shows---although the failure of ROC in a group is then inconclusive for the composition. Verification thus reduces to the individual groups, where it may be settled by Proposition~\ref{prop:trivial-roc} or Proposition~\ref{prop:ptime}, by saturation, or by the procedure of Theorem~\ref{thm:upper}. The decomposition applies even though the synthesis remains global at the template level: ROC is a property of the plant and the relabeling, into which the specification does not enter.

\begin{remark}\label{rem:withingroup}
	The decomposition of Proposition~\ref{prop:comp} operates \emph{across} groups. It does not apply \emph{within} a group, because the agents of one group share a template by construction, so the hypothesis of pairwise disjoint templates fails---as Example~\ref{ex:sef} shows in the sharpest possible way. Verifying ROC of a group with $n_i$ agents therefore still takes place on the synchronous product of the $n_i$ agents, whose size is exponential in $n_i$.
\end{remark}

\section{A Refutation and a Repair}\label{sec:app-ma}
	One structural condition was proposed for the multi-agent setting: if the event sets of the agents inside each group are pairwise disjoint (share-event-free), then, according to~\cite[Prop.~2]{LKL22}, the global plant is ROC by construction, and moreover $L(\mathbf{M})\subseteq R(L(\mathbf{G}))$. The first assertion is incorrect.

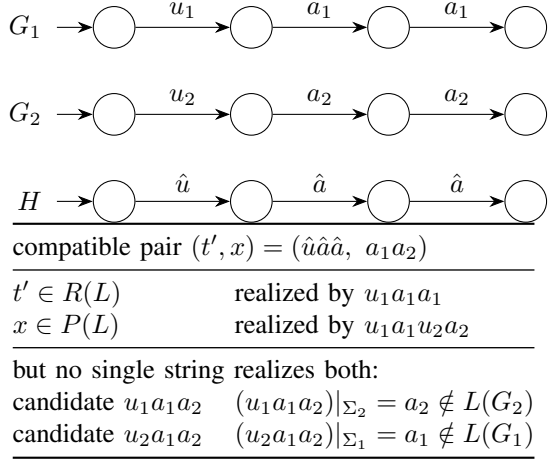
\begin{figure}[t]
\centering
\begin{tikzpicture}[node distance=1.25cm,
    every state/.style={minimum size=5.5mm, inner sep=1pt},
    >={Stealth[length=2mm]}, initial text={}]
  \begin{scope}
    \node[state, initial] (a0) {};
    \node[state, right=of a0] (a1) {};
    \node[state, right=of a1] (a2) {};
    \node[state, right=of a2] (a3) {};
    \path[->] (a0) edge node[above] {$u_1$} (a1)
              (a1) edge node[above] {$a_1$} (a2)
              (a2) edge node[above] {$a_1$} (a3);
    \node[left=0.55cm of a0] {$G_1$};
  \end{scope}
  \begin{scope}[yshift=-1.15cm]
    \node[state, initial] (b0) {};
    \node[state, right=of b0] (b1) {};
    \node[state, right=of b1] (b2) {};
    \node[state, right=of b2] (b3) {};
    \path[->] (b0) edge node[above] {$u_2$} (b1)
              (b1) edge node[above] {$a_2$} (b2)
              (b2) edge node[above] {$a_2$} (b3);
    \node[left=0.55cm of b0] {$G_2$};
  \end{scope}
  \begin{scope}[yshift=-2.30cm]
    \node[state, initial] (h0) {};
    \node[state, right=of h0] (h1) {};
    \node[state, right=of h1] (h2) {};
    \node[state, right=of h2] (h3) {};
    \path[->] (h0) edge node[above] {$\hat u$} (h1)
              (h1) edge node[above] {$\hat a$} (h2)
              (h2) edge node[above] {$\hat a$} (h3);
    \node[left=0.55cm of h0] {$H$};
  \end{scope}
\end{tikzpicture}

\begin{tabular}{@{}ll@{}}
  \toprule
  \multicolumn{2}{@{}l}{compatible pair $(t',x)=(\hat u\hat a\hat a,\ a_1a_2)$}\\
  \midrule
  $t'\in R(L)$ & realized by $u_1a_1a_1$\\
  $x\in P(L)$  & realized by $u_1a_1u_2a_2$\\
  \midrule
  \multicolumn{2}{@{}l}{but no single string realizes both:}\\
  candidate $u_1a_1a_2$ & $(u_1a_1a_2)|_{\Sigma_2}=a_2\notin L(G_2)$\\
  candidate $u_2a_1a_2$ & $(u_2a_1a_2)|_{\Sigma_1}=a_1\notin L(G_1)$\\
  \bottomrule
\end{tabular}
\caption{The two isomorphic machines $G_1,G_2$ of Example~\ref{ex:sef} and their common template $H$. The events $u_1,u_2$ are unobservable and $a_1,a_2$ observable, with $R(u_i)=\hat u$ and $R(a_i)=\hat a$; all states are final. The pair $(\hat u\hat a\hat a,\,a_1a_2)$ is compatible, but a witness would have to be $u_ja_1a_2$ for some $j$, and either choice leaves one machine processing a workpiece it was never set up for. The group is share-event-free, so this contradicts~\cite[Prop.~2]{LKL22}.}
\label{fig:sef}
\end{figure}

\begin{example}\label{ex:sef}
	Consider one group of two isomorphic machines with disjoint alphabets, shown in Fig.~\ref{fig:sef}. Machine $i$ is set up by the unobservable event $u_i$ and then processes up to two workpieces, each by the observable event $a_i$. Its behavior is $\overline{u_i a_i a_i}$, an isomorphic copy of the template behavior $\overline{\hat u\hat a\hat a}$ under $R(a_i)=\hat a$ and $R(u_i)=\hat u$. The plant $L$ is the shuffle of the two behaviors. The pair $(\hat u\hat a\hat a,\,a_1a_2)$ is compatible: $\hat u\hat a\hat a\in R(L)$ via $u_1a_1a_1$ and $a_1a_2\in P(L)$ via $u_1a_1u_2a_2$. However, it is not realized; indeed, a witness must be of the form $u_ja_1a_2$, but $u_1a_1a_2$ restricts machine~2 to $a_2$, a job without a setup, and $u_2a_1a_2$ restricts machine~1 to $a_1$. Neither restriction is in the corresponding behavior. The single unobservable setup prescribed by the template cannot serve both machines, and $R$ is not ROC with respect to $L$ and $P$, although the group is share-event-free.
	\hfill$\diamond$
\end{example}

	The small factory illustrating the approach in~\cite{LKL22} violates ROC as well. The template string $\mathit{start}\cdot\mathit{breakdown}\cdot\mathit{start}\cdot\mathit{finish}$ of the input group, paired with the observation in which the second input machine starts, then the first input machine starts, and then the second machine finishes, is compatible but not realized: a breakdown occurs only after a start and prevents finishing until a repair, so the unobservable breakdown can be attributed neither to the first machine, which has not started yet, nor to the second machine, which subsequently finishes without being repaired.

\subsubsection*{Scope of the refutation, and a repair}
	Two remarks on the scope. First, the refutation is independent of how the template plant $\mathbf{M}$ is formed, that is, of the number of agents entering it: by Definition~\ref{def:roc}, ROC is a property of the plant, the relabeling, and the observation alone, with $t'$ ranging over all of $R(L(\mathbf{G}))$---in this form ROC is used in the proof of~\cite[Thm.~2]{LKL22}. Second, ROC is a sufficient condition: its failure invalidates the maximal-permissiveness \emph{guarantee} of~\cite[Thm.~2]{LKL22}, but does not imply that the scalable supervisor is suboptimal.

	The second assertion of~\cite[Prop.~2]{LKL22}, the inclusion $L(\mathbf{M})\subseteq R(L(\mathbf{G}))$, is untouched by the refutation---but its published proof is also not sound: \cite[Lem.~9]{LKL21} argues via $R(\bigcap_i P_i^{-1}(s_i))=\bigcap_i R(P_i^{-1}(s_i))$, that is, it lets a noninjective morphism commute with intersection, which holds only as ``$\subseteq$''. Lemma~\ref{lem:amalg} repairs the argument, and Corollary~\ref{cor:LMinRLG} states the inclusion with the hypotheses it needs.

\begin{corollary}\label{cor:LMinRLG}
	Let $R_i\colon\Sigma_i\to T_i$ be relabelings with pairwise disjoint templates, where $\Sigma_i=\bigcup_{j\le n_i}\Sigma_{ij}$ and the agent alphabets $\Sigma_{ij}$ within each group are pairwise disjoint, and let every $L(G_{ij})$ be nonempty. Let $\mathbf{G}=\|_{i}\|_{j\le n_i} G_{ij}$, and let $\mathbf{M}=\|_i M_i$, where $M_i=R_i(\|_{j\le k_i}G_{ij})$ with $k_i\le n_i$ is taken as a language over the full template alphabet $T_i$. Then $L(\mathbf{M})\subseteq R(L(\mathbf{G}))$.
\end{corollary}
\begin{proof}
	Let $t\in L(\mathbf{M})$ and put $t_i=t|_{T_i}$; since $M_i$ is declared over $T_i$, we have $t_i\in R_i(\|_{j\le k_i}L(G_{ij}))$, so there is $s_i\in\|_{j\le k_i}L(G_{ij})$ with $R_i(s_i)=t_i$. The languages $L(G_{ij})$ are nonempty and prefix-closed, so $\varepsilon\in L(G_{ij})$; since the agents inside a group have pairwise disjoint alphabets, $s_i|_{\Sigma_{ij}}=\varepsilon\in L(G_{ij})$ for $j>k_i$, and hence $s_i\in\|_{j\le n_i}L(G_{ij})$. By Lemma~\ref{lem:amalg}(1), there is $s$ with $R(s)=t$ and $s|_{\Sigma_i}=s_i$ for every $i$, and hence $s\in L(\mathbf{G})$ and $t\in R(L(\mathbf{G}))$.
\end{proof}

	The hypotheses cannot be dropped. Without the disjointness of the agent alphabets within a group, one group of two agents over the same alphabet $\{c\}$ with $R_1(c)=\hat c$, $L(G_{11})=\overline{\{c\}}$, $L(G_{12})=\{\varepsilon\}$, and $k_1=1$ yields $\hat c\in L(\mathbf{M})$, whereas $L(\mathbf{G})=\{\varepsilon\}$. The convention that $M_i$ is a language over the full alphabet $T_i$ matters as soon as different groups share events: if $M_i$ is taken over the image alphabet $R_i(\bigcup_{j\le k_i}\Sigma_{ij})$ instead, then for $L(G_{11})=\overline{\{x\}}$, $L(G_{12})=\{\varepsilon\}$ over $\Sigma_{12}=\{s\}$, and $L(G_{21})=\overline{\{s\}}$, with $R_1(x)=\hat x$, $R_1(s)=R_2(s)=\hat s$, and $k_1=k_2=1$, the shuffle $R_1(\overline{\{x\}})\,\|\,R_2(\overline{\{s\}})$ of languages over $\{\hat x\}$ and $\{\hat s\}$ contains $\hat s$, whereas $G_{12}$ blocks $s$, so that $L(\mathbf{G})=\overline{\{x\}}$ and $\hat s\notin R(L(\mathbf{G}))$.

\begin{remark}\label{rem:vacuity}
	By Proposition~\ref{prop:trivial-roc}, within a group of at least two agents having observable events, $R$ is never injective on $\So$, so in the presence of unobservable events ROC never holds ``for free''; Example~\ref{ex:sef} and the small factory show that it frequently fails. The positive results of Section~\ref{sec:suff} delimit when it does hold: saturation is precisely the statement that the instantiations of an observable template event are interchangeable in the plant, which is the case for pools of identical, mutually substitutable machines and fails as soon as an unobservable event individualizes an agent, as the setup event does in Example~\ref{ex:sef}.
\end{remark}

\section{Conclusion}\label{sec:concl}
	Relabeling observation consistency is decidable, and verifying it is \PSPACE-complete, already for deterministic plants over alphabets of the smallest possible sizes. Decidability rests on a single structural fact: because a relabeling renames events but erases none, a compatible pair of an observation and a template string has exactly one code over the pairing alphabet, and the condition becomes an inclusion of two nondeterministic automata of quadratic size. The same fact fails for the projection abstractions of hierarchical control, where the corresponding condition is undecidable~\cite{MOC}; the difference is the erasure, not the loss of the identity of events that renaming performs.

	Since the worst case is intractable, the practical value lies in the positive results. In the presence of unobservable events, injectivity of the relabeling on observable events characterizes the relabelings for which ROC holds for every plant, and this injectivity never holds within a group of at least two agents with observable events, so ROC is never automatic in the intended application. Saturation, simulation, and the polynomial-time fragment for deterministic plants with an injective relabeling on unobservable events give tests that do apply, and compositionality reduces the verification of a multi-agent plant with groupwise disjoint templates to its individual groups.

	Finally, the structural condition proposed in the literature to guarantee ROC does not hold, and it fails already for a group of two identical machines each of which is set up by one unobservable event. What fails is exactly what saturation asks for: an unobservable event that individualizes an agent cannot be exchanged for its copy in another agent.

\bibliographystyle{IEEEtran}
\bibliography{roc}

\end{document}